\documentclass[journal]{IEEEtran}
\usepackage{amsfonts}
\usepackage{amssymb}
\usepackage{amsthm}
\usepackage{amsmath,amsfonts,amssymb}
\usepackage[dvips]{graphicx}
\usepackage{verbatim}
\usepackage{setspace}
\usepackage{bm}
\usepackage[ruled,vlined]{algorithm2e}
\usepackage{cite}

\usepackage{changepage}
\usepackage{pdfpages}
\newtheorem{theorem}{Theorem}
\newtheorem{lemma}{Lemma}

\newcommand{\figwidth}{8.8}
\IEEEoverridecommandlockouts

\begin{document}
\title{Joint Power Control and Beamforming for Uplink Non-Orthogonal Multiple Access in 5G Millimeter-Wave Communications}
\author{Lipeng Zhu,
        Jun Zhang,
        Zhenyu Xiao,~\IEEEmembership{Senior Member,~IEEE,}
        Xianbin Cao,~\IEEEmembership{Senior Member,~IEEE,}
        Dapeng Oliver Wu,~\IEEEmembership{Fellow,~IEEE}
        and Xiang-Gen Xia,~\IEEEmembership{Fellow,~IEEE}
%\thanks{This work was supported in part by the National Natural Science Foundation of China (NSFC) under grant Nos. 61571025 and 91538204, and in part by the Foundation for Innovative Research Groups of the National Natural Science Foundation of China under grant No. 61521091.}
\thanks{L. Zhu, J. Zhang, Z. Xiao and X. Cao are with the School of
Electronic and Information Engineering, Beihang University, Beijing 100191, China.}
\thanks{D. O. Wu is with the Department of Electrical and Computer Engineering, University of Florida, Gainesville, FL 32611, USA.}
\thanks{X.-G. Xia is with the Department of Electrical and Computer Engineering, University of Delaware, Newark, DE 19716, USA.}
}


\maketitle

\begin{abstract}
In this paper, we investigate the combination of two key enabling technologies for the fifth generation (5G) wireless mobile communication, namely millimeter-wave (mmWave) communications and non-orthogonal multiple access (NOMA). In particular, we consider a typical 2-user uplink mmWave-NOMA system, where the base station (BS) equips an analog beamforming structure with a single RF chain and serves 2 NOMA users. An optimization problem is formulated to maximize the achievable sum rate of the 2 users while ensuring a minimal rate constraint for each user. The problem turns to be a joint power control and beamforming problem, i.e., we need to find the beamforming vectors to steer to the two users simultaneously subject to an analog beamforming structure, and meanwhile control appropriate power on them. As direct search for the optimal solution of the non-convex problem is too complicated, we propose to decompose the original problem into two sub-problems that are relatively easy to solve: one is a power control and beam gain allocation problem, and the other is an analog beamforming problem under a constant-modulus constraint. The rational of the proposed solution is verified by extensive simulations, and the performance evaluation results show that the proposed sub-optimal solution achieve a close-to-bound uplink sum-rate performance.
\end{abstract}


% \fi
\IEEEpeerreviewmaketitle

\section{Introduction}
\IEEEPARstart{R}{ecently}, the fifth generation (5G) mobile communication has drawn worldwide attention, and the commercial use of 5G is approaching. The fast growth of mobile Internet has propelled 1000-fold data traffic increase by 2020. Apparently, large area capacity is one of the most important requirements of 5G \cite{andrews2014will}. In order to improve the capacity of the 5G network, there are mainly three candidate key technologies, namely extreme densification of cells, millimeter-wave (mmWave) communication and massive multiple-input multiple-output(MIMO) \cite{andrews2014will,niu2015survey,rapp2013mmIEEEAccess}. Indeed, mmWave communication promises a much higher capacity than the legacy low-frequency (i.e., micro-wave band) mobile communications because of abundant frequency spectrum resource. For this reason, it is also considered as a future technology to improve the transmission capacity in the airborne communications, e.g., unmanned aerial vehicle (UAV) communications \cite{xiao2016UAV}.

% rely on efficient MA , OMA is not enough for 5G, mmWave NOMA
%In another way, due to the rapid development of the Internet of Things (IoT), 5G needs to support massive connectivity of users and/or devices to meet the demand for low latency, low-cost devices, and diverse service types. However, conventional orthogonal multiple access (TDMA/FDMA/CDMA) can not meet these requirements. Orthogonal multiple access serves only one user in a resource block while all users share equal resource ignoring the difference of channel condition and  date-rate requirement. In order to avert the waste of resource, non-orthogonal multiple access (NOMA) was proposed.

On the other hand, the non-orthogonal multiple access (NOMA) technique has recently received considerable attention as a promising multiple access technique to be used in 5G mobile communications \cite{ding2014performance,saito2013non,Ding2015Cooperative,Dai2015NOMA5G,Benjebbour2013ConceptNOMA,
UplinkNOMAwithMulti-Antenna,uplinkNOMAMIMO,Uplink_NOMA_in_5G}. In contrast to the conventional orthogonal multiple access (OMA) schemes, NOMA serves multiple users in one orthogonal resource to improve the spectrum efficiency, as well as increase the number of users. The application of NOMA in future mobile networks can meet the requirements to support massive connectivity of users and/or devices and meet the demand for low latency, low-cost devices, and diverse service types \cite{andrews2014will}. By using sophisticated power allocation at the transmitters, as well as successive interference cancellation (SIC) to mitigate multi-user interference at receivers, the number of users and the spectrum efficiency can be significantly improved, especially when the channel conditions of the users are quite different \cite{ding2014performance,saito2013non,Ding2015Cooperative,Dai2015NOMA5G,Benjebbour2013ConceptNOMA,
UplinkNOMAwithMulti-Antenna,Uplink_NOMA_in_5G,uplinkNOMAMIMO}.

In the way to use mmWave communication in 5G cellular, a big challenge is to support a great number of users. Subject to the hardware cost, the number of radio-frequency (RF) chains in an mmWave device is usually much smaller than that of antennas. As a result, the maximal number of users that can be served within one time/frequency/code resource block is very limited, i.e., usually no larger than the number of RF chains \cite{Xia_2011_60GHz_Tech,xia_2008_prac_ante_traning,wang_2009_beam_codebook,xiao2017mmWaveFD,alkhateeb2014mimo,roh2014millimeter,sun2014mimo}. In such a case, NOMA is with significance for mmWave communication to greatly increase the number of users, and meanwhile increase the usage efficiency of the acquired spectrum to support the exponential
traffic growth. Moreover, mmWave communications usually uses the highly directional feature of mmWave propagation, which makes the users' channels (along the same direction) highly correlated and hence facilitates the application of NOMA. For these reasons, we investigate NOMA in mmWave communications (mmWave-NOMA) \cite{Ding2017random,Daill2017} in this paper .

Different from the conventional micro-wave band communications, antenna array is usually adopted in mmWave communications to achieve high array gain to bridge the link budget gap due to the extremely high path loss, which means that beamforming is usually entangled with power control/allocation in mmWave-NOMA. Moreover, since the number of RF chains is usually much smaller than that of antennas in mmWave communications \cite{Xia_2011_60GHz_Tech,xia_2008_prac_ante_traning,wang_2009_beam_codebook,alkhateeb2014mimo,roh2014millimeter,sun2014mimo}, the joint beamforming and power control/allocation in mmWave-NOMA behaves quite differently from that in the conventional MIMO systems, where a fully digital beamforing structure is usually adopted, i.e., the numbers of RF chains and antennas are equal. In \cite{Uplink_NOMA_in_5G}, power control was explored so as to maximize the sum rate with minimal user rate guaranteed, but the power control problem is studied under fixed beam pattern. In \cite{Ding2017random}, the performance of mmWave-NOMA was analyzed by assuming random beamforming with fixed power allocation. In \cite{Daill2017}, the new concept of beamspace MIMO-NOMA with a lens-array hybrid beamforming structure was firstly proposed to use multi-beam forming to serve multiple NOMA users with arbitrary locations, thus the limit that the number of supported users cannot be larger than the number of RF chains can be broken. However, the power allocation problem is studied under fixed beam pattern when lens array is considered. In \cite{xiao2017Downlink_NOMA}, joint power allocation and beamforming was explored in a 2-user downlink mmWave-NOMA scenario with a constant-modulus (CM) phased array. Different from \cite{Ding2017random,Uplink_NOMA_in_5G,Daill2017,xiao2017Downlink_NOMA}, we consider joint power control and beamforming to maximize the sum rate of a 2-user \emph{uplink} mmWave-NOMA system using an analog beamforming structure with a CM phased array. Note that a significant difference between this paper and \cite{xiao2017Downlink_NOMA} is that we consider uplink transmission here, where the user achievable rates and achievable sum rate have different expressions and thus the problem formulation is different from the downlink case in \cite{xiao2017Downlink_NOMA}. As a result, new techniques that are different from those in \cite{xiao2017Downlink_NOMA} are required to formulate and solve the sub-problems. In fact, to the best of our knowledge, there is still no work considering uplink mmWave-NOMA in the literature.

The formulated joint power control and beamforming problem is difficult to solve because it is not convex. Direct search for the optimal solution is too complicated, because the number of variables is large due to the large number of antennas. Hence, we propose to decompose the original problem into two sub-problems which are relatively easy to solve. One sub-problem is a power control and beam gain allocation problem, which can be solved directly with an analytical approach, and the other is a beamforming problem under the CM constraint, which can be converted into a standard convex optimization problem. Extensive performance evaluations are conducted to verify the rational of the proposed solution, and the results show that the proposed sub-optimal solution achieve a close-to-bound uplink sum-rate performance.

The rest of the paper is organized as follows. In Section II, we present the system model and formulate the problem. In Section III, we propose the solution. In Section IV, simulation results are given to demonstrate the performance of the proposed solution, and the paper is concluded lastly in Section V.

Symbol Notation: $a$ and $\mathbf{a}$ denote a scalar variable and a vector, respectively. $(\cdot)^{\rm{*}}$, $(\cdot)^{\rm{T}}$ and $(\cdot)^{\rm{H}}$ denote conjugate, transpose and conjugate transpose, respectively. $|\cdot|$ and $\|\cdot\|$ denote the absolute value and two-norm, respectively. $\mathbb{E}(\cdot)$ denotes the expectation operation. $\mathrm{Re}(\cdot)$ denotes the real part of a complex number. $[\mathbf{a}]_i$ denotes the $i$-th entry of $\mathbf{a}$.

\section{System Model and Problem Formulation}
\subsection{System model}
In NOMA systems, the multi-user interference will increase with the number of users served within one time/frequency/code resource block, which degrades the average rate of each user and increase the average decoding delay \cite{Benjebbour2013ConceptNOMA,Sun2015MIMONOMA,Ding2017random,Ding2017pairing,Ding2016pairing}. For this reason, the number of NOMA users is not large in general. Without loss of generality, we consider an uplink scenario with two users\footnote{The extension to more users will also be discussed later.} in this paper as shown in Fig. \ref{fig:system}, where a base station (BS) equipped with an $N$-element antenna array serves two users with a single antenna\footnote{In the case that the users also use an antenna array, Tx beamforming can be done first. Then the transmission processing at each user can be seen equivalent to a single-antenna transmitter, and the proposed solution in this paper can be used.}. At the BS, each antenna branch has a phase shifter and a low-noise amplifier (LNA) to drive the antenna. Generally, all the LNAs have the same scaling factor. Thus, the beamforming vector, i.e., the antenna weight vector (AWV), has constant-modulus (CM) elements. User $i$ ($i=1,~2$) transmits a signal $s_{i}$ to the BS, where $\mathbb{E}(\left | s_{i} \right |^{2})=1$, with transmission power $p_{i}$. The total transmission power of each User is restricted to $P$. With 2-user NOMA, signals $s_{1}$ and $s_{2}$ are superimposed at the BS as
\begin{equation}
y=\mathbf{h}_{1}^{\rm{H}}\mathbf{w}\sqrt{p_{1}}s_{1}+\mathbf{h}_{2}^{\rm{H}}\mathbf{w}\sqrt{p_{2}}s_{2}+\mathbf{n}^{\rm{H}}\mathbf{w}
\end{equation}
where ${\mathbf{h}}_{i}$ is channel response vectors between User $i$ and the BS, ${\mathbf{w}}$ denotes a CM beamforming vector with $|[{\mathbf{w}}]_k|=\frac{1}{\sqrt{N}}$ for $k=1,2,...,N$, and $\mathbf{n}$ is an $N$-dimention vector that denotes the Gaussian white noises of $N$-antenna at the BS with power $\sigma^{2}$.
\begin{figure}[t]
\begin{center}
  \includegraphics[width=\figwidth cm]{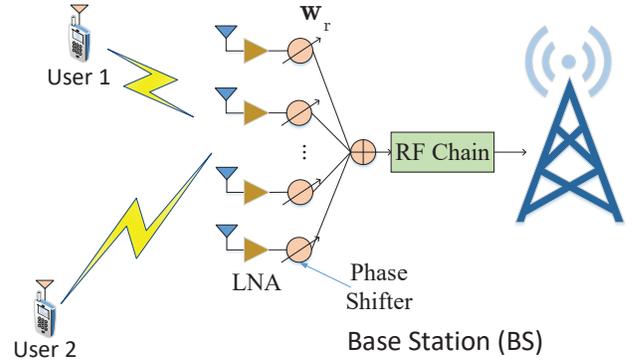}
  \caption{Illustration of an mmWave mobile cell, where one BS with $N$ antennas serves multiple users with one single antenna.}
  \label{fig:system}
\end{center}
\end{figure}

The channel between User $i$ and the BS is an mmWave channel. Subject to limited scattering in the mmWave band, multipath is mainly caused by reflection. As the number of the multipath components (MPCs) is small in general, the mmWave channel has directionality and appears spatial sparsity in the angle domain. Different MPCs have different angles of arrival (AoAs). Without loss of generality, we adopt the directional mmWave channel model assuming a uniform linear array (ULA) with a half-wavelength antenna space. Then an mmWave channel can be expressed as \cite{peng2015enhanced,wang2015multi,Lee2014exploiting,Gao2016ChannelEst,xiao2016codebook,xiao2017codebook}
\begin{equation} \label{eq_oriChannel}
\bar{\mathbf{h}}_{i}=\sum_{\ell=1}^{L_i}\lambda_{i,\ell}\mathbf{a}(N,\Omega_{i,\ell})
\end{equation}
where $\lambda_{i,\ell}$, $\Omega_{i,\ell}$ are the complex coefficient and cos(AoA) of the $\ell$-th MPC of the channel vector for User $i$, respectively, $L_i$ is the total number of MPCs for User $i$, ${\bf{a}}(\cdot)$ is a steering vector function defined as
\begin{equation} \label{eq_steeringVCT}
\mathbf{a}(N,\Omega)=[e^{j\pi0\Omega},e^{j\pi1\Omega},e^{j\pi2\Omega},\cdot\cdot\cdot,e^{j\pi(N-1)\Omega}]
\end{equation}
which depends on the array geometry. Let $\theta_{i,\ell}$ denote the real AoA of the $\ell$-th MPC for User $i$, then we have $\Omega_{i,\ell}=\cos(\theta_{i,\ell})$. Therefore, $\Omega_{i,\ell}$  is within the range $[-1, 1]$. For convenience and without loss of generality, in the rest of this paper, $\Omega_{i,\ell}$ is also called AoA.

For each user, the BS would perform beamforming toward the angle direction along the AoA of the strongest MPC to achieve a high array gain. In general, if there is no blockage between the BS and a user, the line-of-sight (LOS) component will be adopted for beamforming, as it has a much higher strength than the non-LOS (NLOS) components. If the LOS component is blocked, the strongest NLOS path would be selected for beamforming. Since the mmWave channel is spatially sparse, we can obtain an effective channel model for the original channel model \eqref{eq_oriChannel} as
\begin{equation} \label{eq_effChannel}
\mathbf{h}_{i}=\lambda_{i}\mathbf{a}(N,\Omega_{i})
\end{equation}
where $\lambda_{i}=\lambda_{i,m_i}$ and $\Omega_{i}=\Omega_{i,m_i}$. Here $m_i$ denotes the index of the strongest MPC for User $i$. Since the effective channel model \eqref{eq_effChannel} is simpler, we adopt it in the derivation and analysis in this paper, while in the performance evaluations we also consider the original channel model in \eqref{eq_oriChannel}. Without loss of generality, we assume $|\lambda_{1}|\geq |\lambda_{2}|$, which means that the channel gain of User 1 is better.

\subsection{Decoding order}
In the conventional uplink NOMA with single-antenna BS and users, usually the information of the user with a higher channel gain is decoded first to maximize the sum rate. In contrast, in mmWave-NOMA, the decoding order depends on both channel gain and beamforming gain. Thus, there are two cases for the 2-user uplink mmWave-NOMA system.

\emph{Case 1:} $s_1$ is decoded first. Then $s_2$ is decoded after subtracting the signal component of $s_1$. With this decoding method, the achievable rates of User $i~(i=1,2)$, denoted by $R_i$ are represented as
\begin{equation}
\left\{\begin{aligned}
R_{1}^{(1)}&=\log_{2}(1+ \frac{\left |\mathbf{h}_{1}^{\rm{H}}\mathbf{w} \right |^{2}p_{1}}{\left |\mathbf{h}_{2}^{\rm{H}}\mathbf{w} \right |^{2}p_{2}+\sigma^{2}}) \\
R_{2}^{(1)}&=\log_{2}(1+ \frac{\left |\mathbf{h}_{2}^{\rm{H}}\mathbf{w} \right |^{2}p_{2}}{\sigma^{2}})
\end{aligned}\right.
\end{equation}

\emph{Case 2:} $s_2$ is decoded first. Then $s_1$ is decoded after subtracting the signal component of $s_2$. With this decoding method, the achievable rates of User $i~(i=1,2)$, denoted by $R_i$ are represented as
\begin{equation}
\left\{\begin{aligned}
R_{1}^{(2)}&=\log_{2}(1+ \frac{\left |\mathbf{h}_{1}^{\rm{H}}\mathbf{w} \right |^{2}p_{1}}{\sigma^{2}}) \\
R_{2}^{(2)}&=\log_{2}(1+ \frac{\left |\mathbf{h}_{2}^{\rm{H}}\mathbf{w} \right |^{2}p_{2}}{\left |\mathbf{h}_{1}^{\rm{H}}\mathbf{w} \right |^{2}p_{1}+\sigma^{2}})
\end{aligned}\right.
\end{equation}

The expressions of the achievable sum rate of under different decoding orders are identical, which can be calculated directly as
\begin{equation} \label{eq_sum_rate}
R_{1}+R_{2}=\log_{2}(1+ \frac{\left |\mathbf{h}_{1}^{\rm{H}}\mathbf{w} \right |^{2}p_{1}+\left |\mathbf{h}_{2}^{\rm{H}}\mathbf{w} \right |^{2}p_{2}}{\sigma^{2}})
\end{equation}

\subsection{Problem Formulation}
An immediate and basic problem is how to maximize the achievable sum rate of the two users provided that the channel is known a priori. It is clear that if there are no minimal rate constraints for the two users, the achievable sum rate can be maximized by transmitting signal with the maximal user power at each user and meanwhile allocating all the beamforming gain toward User 1, whose channel gain is better. However, when there are minimal rate constraints for the two users, the power control intertwines with the beamforming design, which makes the problem complicated under the system setup. In this paper, we intend to address this problem, which is formulated by

\begin{equation}\label{eq_problem}
\begin{aligned}
\mathop{\mathrm{Maximize}}\limits_{p_1,p_2,{\mathbf{w}}}~~~~ &R_{1}+R_{2}\\
\mathrm{Subject~ to}~~~ &R_{1} \geq r_{1}\\
&R_{2} \geq r_{2} \\
&0 \leq p_{1},p_{2} \leq P \\
&|[{\mathbf{w}}]_k|=\frac{1}{\sqrt{N}},~k=1,2,...,N\\
\end{aligned}
\end{equation}
where $r_i$ denotes the minimal rate constraint for User $i$, $|[{\bf{w}}]_k|=\frac{1}{\sqrt{N}}$ is the CM constraint due to using the phase shifters in each antenna branch at the BS. Note that the expressions of $R_1$ and $R_2$ are different for different decoding orders. In Case $i$ ($i=1,2$), $R_1=R_1^{(i)}$ and $R_2=R_2^{(i)}$. However, we will prove in the next section that the achievable sum rate of Case 1 is better than that of Case 2. The above Problem \eqref{eq_problem} is challenging, not only due to the non-convex constraints and the objective function, but also due to that the parameters to be optimized are entangled with each other.

\section{Solution of the Problem}
Clearly, directly solving Problem \eqref{eq_problem} by using the existing optimization tools is infeasible, because the problem is non-convex and may not be converted to a convex problem with simple manipulations. On the other hand, to directly search the optimal solution is also computationally prohibitive because the dimension is $(N+2)$, where $N$ is large in general. In this section, we propose a suboptimal solution to this problem. The
basic idea is to decompose the original problem \eqref{eq_problem} into two sub-problems which are relatively easy to solve, and then we solve them one by one.

\subsection{Problem Decomposition}

Since power control intertwines with beamforming under the CM constraint, we first try to decompose them.
Let $c_{1}=\left|{{\bf{h}}}_{1}^{\rm{H}}\mathbf{w}\right|^{2}$ and $c_{2}=\left|{{\bf{h}}}_{2}^{\rm{H}}\mathbf{w}\right|^{2}$ denote the beam gains for User 1 and User 2, respectively. We directly give the following lemma which has been proven in \cite{xiao2017Downlink_NOMA}:
\begin{lemma} With the ideal beamforming, the beam gains satisfy
\begin{equation}\label{ideal_condition}
\frac{c_{1}}{\left|\lambda_{1}\right|^{2}}+\frac{c_{2}}{\left|\lambda_{2}\right|^{2}}=N
\end{equation}
where $N$ is the number of antennas.
\end{lemma}

Based on Lemma 1, we can rewrite Problem \eqref{eq_problem} with the beamforming gains. Since the sum rate expressions are different for different decoding orders, the problems are also different for different cases. Problem \eqref{eq_problem} under the two cases can be re-described as\\
\emph{Case 1:}
\begin{equation} \label{eq_problem_sub1_case1}
\begin{aligned}
\mathop{\mathrm{Maximize}}\limits_{p_1,p_2,c_1,c_2}~~~~&\log_{2}(1+ \frac{c_{1}p_{1}+c_{2}p_{2}}{\sigma^{2}})\\
\mathrm{Subject~to} ~~~ &\log_{2}(1+ \frac{c_{1}p_{1}}{c_{2}p_{2}+\sigma^{2}}) \geq r_{1}\\
&\log_{2}(1+ \frac{c_{2}p_{2}}{\sigma^{2}}) \geq r_{2} \\
&0 \leq p_{1},p_{2} \leq P \\
&\frac{c_{1}}{\left|\lambda_{1}\right|^{2}}+\frac{c_{2}}{\left|\lambda_{2}\right|^{2}}=N\\
\end{aligned}
\end{equation}
\emph{Case 2:}
\begin{equation} \label{eq_problem_sub1_case2}
\begin{aligned}
\mathop{\mathrm{Maximize}}\limits_{p_1,p_2,c_1,c_2}~~~~&\log_{2}(1+ \frac{c_{1}p_{1}+c_{2}p_{2}}{\sigma^{2}})\\
\mathrm{Subject~to} ~~~ &\log_{2}(1+ \frac{c_{1}p_{1}}{\sigma^{2}}) \geq r_{1}\\
&\log_{2}(1+ \frac{c_{2}p_{2}}{{c_{1}p_{1}+\sigma^{2}}}) \geq r_{2} \\
&0 \leq p_{1},p_{2} \leq P \\
&\frac{c_{1}}{\left|\lambda_{1}\right|^{2}}+\frac{c_{2}}{\left|\lambda_{2}\right|^{2}}=N\\
\end{aligned}
\end{equation}
where $|{\bf{h}}_i^{\rm{H}}{\bf{w}}|^2$ is replaced by the beam gain $c_i$ ($i=1,2$). The CM constraint is not involved in Problems \eqref{eq_problem_sub1_case1} and \eqref{eq_problem_sub1_case2}, but will be considered in the following beamforming sub-problem. It is worthy to note that the objective functions are uniform under different decoding orders, which is distinguishing with the downlink scenario in \cite{xiao2017Downlink_NOMA}. For this reason, the optimal decoding order can be uniquely determined, which will be shown in the next subsection.

Next, we formulate the beamforming problem, i.e., to design $\bf{w}$ such that $|{\bf{h}}_i^{\rm{H}}{\bf{w}}|^2=c_i$ ($i=1,2$) under the CM constraint, which is formulated as follows:
\begin{equation} \label{eq_problem_sub2}
\begin{aligned}
&~\mathbf{w}\in \mathbb{C}^{N}\\
\mathrm{Subject~to}~~~&\left|{{\mathbf{h}}}_{1}^{\rm{H}}\mathbf{w}\right|^{2}=c_1\\
&\left|{{\mathbf{h}}}_{2}^{\rm{H}}\mathbf{w}\right|^{2} = {c_{2}} \\
&~|[{\mathbf{w}}]_k|=\frac{1}{\sqrt{N}},~k=1,2,...,N
\end{aligned}
\end{equation}

With the above manipulations, Problem \eqref{eq_problem} is decomposed into Problems \eqref{eq_problem_sub1_case1} and \eqref{eq_problem_sub2}, which are independent power control and beam gain allocation and beamforming sub-problems. Although the original problem is hard to solve, the two sub-problems are relatively easy to solve. Next, we will first solve Problem \eqref{eq_problem_sub1_case1}, and obtain the optimal solution $\{c_{1}^{\star},~c_{2}^{\star},~p_{1}^{\star},~p_{2}^{\star}\}$ of \eqref{eq_problem_sub1_case1}. Then $c_{2}^{\star}$ is used as the gain constraints in Problem \eqref{eq_problem_sub2}. We solve Problem \eqref{eq_problem_sub2} and obtain an appropriate ${\bf{w}}^\circ$. Although the obtained solution $\{p_{1}^{\star}, p_{2}^{\star}, {\bf{w}}^\circ\}$ is not globally optimal, the achieved sum rate performance is close to the upper bound, as it will be shown later in Section V. Next, we show how to solve the two sub-problems.

\subsection{Solution of the Power Control and Beam Gain Allocation Sub-Problem}

It is noteworthy that Problem \eqref{eq_problem_sub1_case1} and Problem \eqref{eq_problem_sub1_case2} are similar to each other, which means that if an approach can be used to solve one of them, it can also be used to solve the other. In fact, it will be shown later that the optimal sum rate of Case 1 is better than that of Case 2. For this reason, we just show the solution of Problem \eqref{eq_problem_sub1_case1} in detail. We first figure out the optimal $\{p_{1}^{\star}, p_{2}^{\star}\}$ in this subsection, and then, the optimal beam gains $\{c_{1}^{\star}, c_{2}^{\star}\}$.

\begin{lemma} With the ideal beamforming, the optimal transmission power is
\begin{equation}\label{optimal_power}
\left\{\begin{aligned}
&p_{1}^{\star}=P\\
&p_{2}^{\star}=P
\end{aligned}\right.
\end{equation}
\end{lemma}

\begin{proof}
Suppose the optimal solution of Problem \eqref{eq_problem_sub1_case1} is $p_{1}=p_{1}^{\star},~p_{2}=p_{2}^{\star},~ c_{1}=c_{1}^{\star},~c_{2}=c_{2}^{\star}$. With the optimal solution, the optimal user rates are $R_{1}=R_{1}^{\star}$ and $R_{2}=R_{2}^{\star}$, respectively.

Assume $p_{1}^{\star}<P$. We consider the parameter settings $p_{1}=P>p_{1}^{\star},~p_{2}=p_{2}^{\star},~ c_{1}=c_{1}^{\star},~c_{2}=c_{2}^{\star}$. Then we have
\begin{equation}
\left\{\begin{aligned}
&R_{1}=\log_{2}(1+ \frac{c_{1}^{\star}P}{c_{2}^{\star}p_{2}^{\star}+\sigma^{2}})>R_{1}^{\star} \geq r_{1}\\
&R_{2}=\log_{2}(1+ \frac{c_{2}^{\star}p_{2}^{\star}}{\sigma^{2}})=R_{2}^{\star} \geq r_{2}\\
&R_{1}+R_{2}>R_{1}^{\star}+R_{2}^{\star}
\end{aligned}\right.
\end{equation}
which means that the rate constraints are all satisfied while the value of the objective function becomes greater. Hence, the assumption of $p_{1}^{\star}<P$ does not hold. We have $p_{1}^{\star}=P$.

Analogously, we assume $p_{2}^{\star}<P$. We consider the parameter settings
\begin{equation}
\left\{\begin{aligned}
&p_{2}=P>p_{2}^{\star}\\
&c_{2}=\frac{c_{2}^{\star}p_{2}^{\star}}{p_2}=\frac{c_{2}^{\star}p_{2}^{\star}}{P}<c_{2}^{\star}\\
&p_{1}=p_{1}^{\star}\\
&c_{1}=\left|\lambda_{1}\right|^{2}(N-\frac{c_{2}}{\left|\lambda_{2}\right|^{2}})>c_{1}^{\star}
\end{aligned}\right.
\end{equation}

With these parameter settings, we have $c_{2}p_{2}=c_{2}^{\star}p_{2}^{\star}$, then
\begin{equation}
\left\{\begin{aligned}
R_{1}&=\log_{2}(1+ \frac{c_{1}p_{1}}{c_{2}p_{2}+\sigma^{2}})\\
&=\log_{2}(1+ \frac{c_{1}p_{1}^{\star}}{c_{2}^{\star}p_{2}^{\star}+\sigma^{2}})>R_{1}^{\star} \geq r_{1}\\
R_{2}&=\log_{2}(1+ \frac{c_{2}p_{2}}{\sigma^{2}})=\log_{2}(1+ \frac{c_{2}^{\star}p_{2}^{\star}}{\sigma^{2}})=R_{2}^{\star} \geq r_{2}\\
R_{1}&+R_{2}>R_{1}^{\star}+R_{2}^{\star}\\
\end{aligned}\right.
\end{equation}
which means that the rate constraints are all satisfied while the value of the objective function becomes greater. Hence, the assumption of $p_{2}^{\star}<P$ does not hold. We have $p_{2}^{\star}=P$.

With the above analyses, the value of the objective function can always increase in the feasible domain when increasing $p_{1}$ or $p_{2}$. Hereto, the optimal values of $p_{1}$ and $p_{2}$ are
\begin{equation}\label{optimal_power}
\left\{\begin{aligned}
&p_{1}^{\star}=P\\
&p_{2}^{\star}=P
\end{aligned}\right.
\end{equation}
\end{proof}

According to Lemma 1 and Lemma 2, we have $p_{1}=P,~p_{2}=P,~c_{1}=\left|\lambda_{1}\right|^{2}(N-\frac{c_{2}}
{\left|\lambda_{2}\right|^{2}})$. Substituting them into Problem \eqref{eq_problem_sub1_case1}, there is only one independent variable $c_{2}$ now. Hence, we can transform the problem as\\
\begin{equation}\label{eq_problem2_case1}
\begin{aligned}
\mathop{\mathrm{Maximize}}\limits_{c_{2}}~~~~ &\log_{2}(1+ \frac{(\left|\lambda_{1}\right|^{2}N-(\frac{\left|\lambda_{1}\right|^{2}}{\left|\lambda_{2}\right|^{2}}-1)c_{2})P}{\sigma^{2}})\\
\mathrm{Subject~ to}~~~ &\log_{2}(1+ \frac{\left|\lambda_{1}\right|^{2}(N-\frac{c_{2}}
{\left|\lambda_{2}\right|^{2}})P}{c_{2}P+\sigma^{2}}) \geq r_{1}\\
&\log_{2}(1+ \frac{c_{2}P}{\sigma^{2}}) \geq r_{2} \\
\end{aligned}
\end{equation}

As $|\lambda_1| \geq |\lambda_2|$, we have $-(\frac{\left|\lambda_{1}\right|^{2}}{\left|\lambda_{2}\right|^{2}}-1) \leq 0$. The objective function is monotonically decreasing for $c_{2}$, so the infimum of $c_{2}$ is optimal. Furthermore, $R_1$ is decreasing for $c_2$ and $R_2$ is increasing for $c_2$. The lowerbound of $c_{2}$ is depended on the second constraint $R_2 \geq r_2$ of Problem \eqref{eq_problem2_case1}.
\begin{equation}
\begin{aligned}
\log_{2}(1+ \frac{c_{2}P}{\sigma^{2}}) \geq r_{2} \Rightarrow c_{2} \geq \frac{(2^{r_{2}}-1)\sigma^{2}}{P}
\end{aligned}
\end{equation}

Hereto, we have solved the power control and beam gain allocation sub-problem in Case 1, i.e., Problem \eqref{eq_problem_sub1_case1}. As we have mentioned before, Problem \eqref{eq_problem_sub1_case1} and Problem \eqref{eq_problem_sub1_case2} are similar to each other, which means Problem \eqref{eq_problem_sub1_case2} can also be solved by the above method. We give the following theorem to compare the optimal solution between these two problems.

\begin{theorem} The maxima of the objective function in the power control and beam gain allocation sub-problem under Case 1 is larger than that under Case 2.
\end{theorem}

\begin{proof}
See Appendix A.
\end{proof}

Theorem 1 shows the comparison of the optimal solutions between Case 1 and Case 2, which proves that the optimal order is to decode the signal of User 1 first, i.e., the one with higher channel gain. Consequently, the optimal values of $\left|{{\mathbf{h}}}_{2}^{\rm{H}}\mathbf{w}\right|^{2}$ and $\left|{{\mathbf{h}}}_{1}^{\rm{H}}\mathbf{w}\right|^{2}$ are
\begin{equation}\label{optimal_beam}
\left\{\begin{aligned}
&c_{2}^{\star}=\frac{(2^{r_{2}}-1)\sigma^{2}}{P}\\
&c_{1}^{\star}=\left|\lambda_{1}\right|^{2}(N-\frac{c_{2}^{\star}}
{\left|\lambda_{2}\right|^{2}})\\
\end{aligned}\right.
\end{equation}

Hereto, we have solved the power control and beam gain allocation sub-problem, i.e., we have found the optimal solution of Problem \eqref{eq_problem_sub1_case1}/\eqref{eq_problem_sub1_case2} and proved that the the optimal solution of Problem \eqref{eq_problem_sub1_case1} is better than that of Problem \eqref{eq_problem_sub1_case2}, which means that decoding $s_1$ first is optimal. As the optimal solution $\{c_1^\star,c_2^\star,p_1^\star,p_2^\star\}$ is obtained under the assumption of the ideal beamforming, i.e., we assume Lemma 1 holds. However, $\{c_1^\star,c_2^\star,p_1^\star,p_2^\star\}$ may not be an optimal solution of the original problem, i.e., Problem \eqref{eq_problem}, because a beamforming vector with beam gains $\{c_1^\star,c_2^\star\}$ may not be found under the CM constraint. Hence, we say the optimal achievable sum rate of Problem \eqref{eq_problem_sub1_case1} is an upper bound of that of the original problem.

\subsection{Solution of the Beamforming Sub-Problem}
In this subsection, we solve the beamforming sub-problem, i.e., we solve Problem \eqref{eq_problem_sub2} to design an appropriate ${\bf{w}}$ to realize the user beam gains $c_{1}^\star$ and $c_{2}^\star$. However, as we have mentioned before, the beamforming vector with beam gains $\{c_1^\star,c_2^\star\}$ may not be found because of the sidelobe in beam pattern. Proper relaxation should be adopted to obtain the appropriate ${\bf{w}}$ in Problem \eqref{eq_problem_sub2}.
On one hand, the optimal value of $c_2$ is the lowerbound, thus there should be an constraint $\left|{{\mathbf{h}}}_{2}^{\rm{H}}\mathbf{w}\right|^2 \geq c_2^\star$, otherwise the constraint $R_2 \geq r_2$ can not be feasible. On the other hand, we have proved the objective function is monotonically decreasing for $c_{2}$ with the ideal beamforming assumption in the previous subsection. In other words, the objective function is monotonically increasing for $c_{1}$. The maximization of $\left|{{\mathbf{h}}}_{1}^{\rm{H}}\mathbf{w}\right|^2$ is equivalent to the maximization of achievable sum rate. Hence, Problem \eqref{eq_problem_sub2} can be relaxed as
\begin{equation} \label{eq_problem_beamforming1}
\begin{aligned}
\mathop{\mathrm{Maximize}}\limits_{{\mathbf{w}}}~~~~&\left|{{\mathbf{h}}}_{1}^{\rm{H}}\mathbf{w}\right|^2\\
\mathrm{Subject~to}~~~&\left|{{\mathbf{h}}}_{2}^{\rm{H}}\mathbf{w}\right|^2 \geq c_{2}^{\star} \\
&|[{\mathbf{w}}]_k|=\frac{1}{\sqrt{N}},~k=1,2,...,N
\end{aligned}
\end{equation}
Define $g=\sqrt{\frac{c_{2}^{\star}}{\left|\lambda_{2}\right|^{2}}}$, the problem above can be rewritten as
\begin{equation} \label{eq_problem_beamforming2}
\begin{aligned}
\mathop{\mathrm{Maximize}}\limits_{{\mathbf{w}}}~~~~&\left|{{\mathbf{a}}}_{1}^{\rm{H}}\mathbf{w}\right|\\
\mathrm{Subject~to}~~~&\left|{{\mathbf{a}}}_{2}^{\rm{H}}\mathbf{w}\right| \geq g \\
&|[{\mathbf{w}}]_k|=\frac{1}{\sqrt{N}},~k=1,2,...,N
\end{aligned}
\end{equation}
where ${\mathbf{a}}_{i}\triangleq \mathbf{a}(N,\Omega_i)$ for $i=1,2$.

Problem \eqref{eq_problem_beamforming2} is also non-convex. The problem is still difficult to
solve due to the equality constraints. Therefore, we relax the equality constraints $|[{\mathbf{w}}]_k|=\frac{1}{\sqrt{N}}$ with inequality constraints $|[{\mathbf{w}}]_k| \leq \frac{1}{\sqrt{N}}$, which is convex. We reformulate the beamforming problem as
\begin{equation} \label{eq_problem_beamforming3}
\begin{aligned}
\mathop{\mathrm{Maximize}}\limits_{{\mathbf{w}}}~~~~&\left|{{\mathbf{a}}}_{1}^{\rm{H}}\mathbf{w}\right|\\
\mathrm{Subject~to}~~~&\left|{{\mathbf{a}}}_{2}^{\rm{H}}\mathbf{w}\right| \geq g \\
&|[{\mathbf{w}}]_k| \leq \frac{1}{\sqrt{N}},~k=1,2,...,N
\end{aligned}
\end{equation}

\begin{theorem} If $\mathbf{w}_{0}$ is the optimal solution of Problem \eqref{eq_problem_beamforming3}, then, $|[{\mathbf{w}_{0}}]_k|=\frac{1}{\sqrt{N}},~k=1,2,...,N$.
\end{theorem}

\begin{proof}
See Appendix B.
\end{proof}

According to Theorem 2, Problem \eqref{eq_problem_beamforming2} is equivalent to Problem \eqref{eq_problem_beamforming3}. It is clear that an arbitrary phase rotation can be added to the vector $\mathbf{w}$ in Problem \eqref{eq_problem_beamforming3} without affecting the beam gains. Thus, if $\mathbf{w}$ is optimal, so is $\mathbf{w}e^{j\phi}$, where $\phi$ is an arbitrary phase within $[0,2\pi)$. Without loss of generality, we may then choose $\phi$ so that ${\mathbf{a}}_{1}^{\rm{H}}\mathbf{w}$ is real and non-negative. Problem \eqref{eq_problem_beamforming3} is tantamount to
\begin{equation} \label{eq_problem_beamforming4}
\begin{aligned}
\mathop{\mathrm{Maximize}}\limits_{{\mathbf{w}}}~~~~&{{\mathbf{a}}}_{1}^{\rm{H}}\mathbf{w}\\
\mathrm{Subject~to}~~~~&\left|{{\mathbf{a}}}_{2}^{\rm{H}}\mathbf{w}\right| \geq g \\
&|[{\mathbf{w}}]_k| \leq \frac{1}{\sqrt{N}},~k=1,2,...,N
\end{aligned}
\end{equation}

Problem \eqref{eq_problem_beamforming4} is still not convex because of the absolute value operation in the first constraint. Thus, we can split it into a serial of convex optimization problems, i.e., we assume different phases for ${{\mathbf{a}}}_{2}^{\rm{H}}\mathbf{w}$ and obtain $M$ convex problems
\begin{equation} \label{eq_problem_beamforming5}
\begin{aligned}
\mathop{\mathrm{Maximize}}\limits_{{\mathbf{w}}}~~~~&{{\mathbf{a}}}_{1}^{\rm{H}}\mathbf{w}\\
\mathrm{Subject~to}~~~~&\mathrm{Re}({{\mathbf{a}}}_{2}^{\rm{H}}\mathbf{w}e^{2\pi j\frac{m}{M}}) \geq g \\
&|[{\mathbf{w}}]_k| \leq \frac{1}{\sqrt{N}},~k=1,2,...,N
\end{aligned}
\end{equation}
where $M$ is the number of total candidate phases $(m=1,2,\cdots,M)$. Each of these $M$ problems can be efficiently solved by using standard convex optimization tools. We select the solution with the maximal objective among the $M$ optimal solutions as the final solution $\mathbf{w}^\circ$.

Hereto, we have obtained a sub-optimal solution of the original problem \eqref{eq_problem}, i.e., $\{p_{1}^{\star},p_{2}^{\star},\mathbf{w}^\circ\}$. As we have have some relaxations in the sub-problems, it is not sure that whether the solution $\{p_{1}^{\star},p_{2}^{\star},\mathbf{w}^\circ\}$ obtained from the sub-problems \eqref{eq_problem2_case1} and \eqref{eq_problem_beamforming1} is located in the feasible region of Problem \eqref{eq_problem}. The answer is yes and we give the following theorem to demonstrate it.

\begin{theorem} If the feasible region of Problem \eqref{eq_problem} is not empty, then $\{p_{1}^{\star},p_{2}^{\star},\mathbf{w}^\circ\}$ is a solution of Problem \eqref{eq_problem}.
\end{theorem}

\begin{proof}
See Appendix C.
\end{proof}

\subsection{Generalization to More-User Case}
Although in this paper we adopt a two-user uplink mmWave-NOMA model, the basic idea of decomposing the original problem into two sub-problems also applies to a more-user uplink mmWave-NOMA system. Based on Lemma 1, the original problem with more users can be decomposed. Lemma 2 is still workable in the more-user case, and the optimal power control is $\{p_i=P,~i=1,2,\cdots K\}$, where $K$ is the number of users. Moreover, the beam gains of lower channel gain users should be set to just satisfy the rate constraints, meanwhile the beam gain of the user with the best channel gain should be maximized in the beamforming sub-problem. However, since the proposed method to solve the 2-user beamforming sub-problem needs to search over $M$ possible phases for one user (see \eqref{eq_problem_beamforming5}), in a $K$-user case, we need to search over $M^{K-1}$ possible phases for $(K-1)$ users. In brief, if the number of active users, i.e., $K$, is not large, the idea of problem decomposition is still applicable to solve the original problem. However, when $K$ is too large, the proposed solution may become not appropriate due to high complexity.

%It is noteworthy that the problem in this paper is formulated in a single resource block. The number of users that perform NOMA can not be too large, otherwise the receiver complexity will increase to a great extent. More users can transmit information on orthogonal resource, then the number of users can be increased. Another way to increase the number of users is to use a hybrid beamforming structure with multiple RF chains, such that the number of users can be increased by $N_{\mathrm{RF}}$ times, where $N_{\mathrm{RF}}$ is the number of RF chains.

Fortunately, based on the proposed solution, there are many other ways to support more users. For instance, one method is to combine with the OMA strategies to manyfold increase the number of users, or to use a hybrid beamforming structure with multiple RF chains, such that the number of users can be increased by $N_{\rm{RF}}$ times, where $N_{\rm{RF}}$ is the number of RF chains. Another method is to still use an analog beamforming structure and shape a few beams. The difference is that each beam steer towards a group of users rather than only one user in this paper. In such a case, we need to consider beam gain allocation between different user groups and power allocation within each user group. This topic will be studied in detail in our future work.

\section{Performance Evaluations}

\begin{figure*}[t]
\begin{center}
  \includegraphics[width=19 cm]{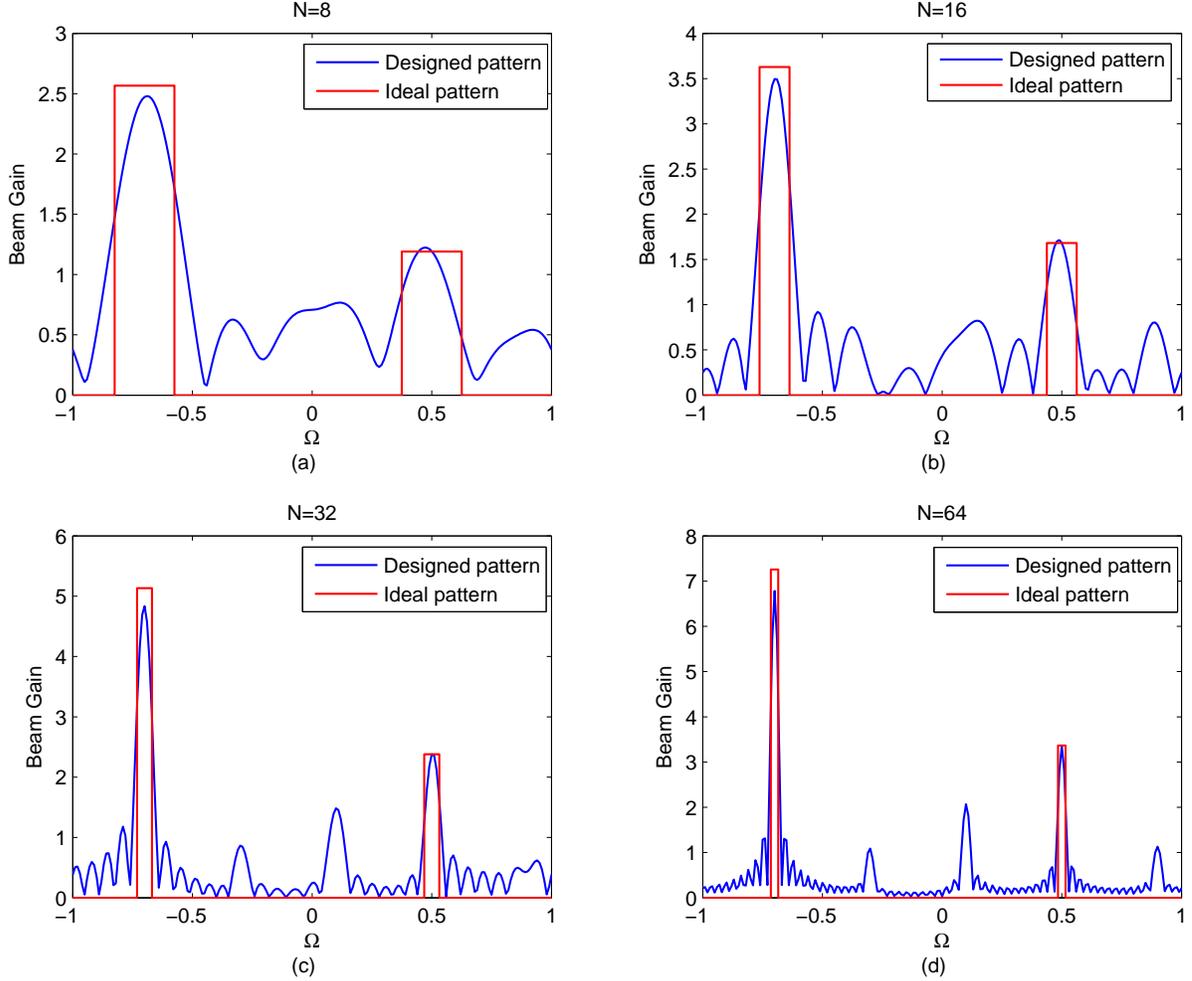}
  \caption{Comparison between the ideal beam pattern and the designed beam pattern}
  \label{fig:beam_pattern}
\end{center}
\end{figure*}
In this section, we evaluate the performance of the proposed joint power control and beamforming method. As aforementioned, the joint problem has been decomposed into two sub-problems, namely the power control and beam gain allocation sub-problem and the beamforming sub-problem. For the power control and beam gain allocation sub-problem, we find the optimal solution; while for the beamforming sub-problem, we find a sub-optimal solution. Hence, we start from the performance evaluation of the beamforming phase.

To compare the ideal beam pattern with the designed beam pattern obtained by solving Problem \eqref{eq_problem_beamforming2}, we assume $|\lambda_1|=0.9$, $|\lambda_2|=0.4$, $\Omega_1=-0.7$, $\Omega_2=0.5$. The desired beam gains are $c_1^\star=2N/3$ and $c_2^\star=(N-c_1^\star/|\lambda_1|^2)|\lambda_2|^2$, where $N$ is the number of antennas at the BS. $M$ in \eqref{eq_problem_beamforming5} is set to 20 in this simulation as well as the following simulations, which is large enough to obtain the best solution. Fig. \ref{fig:beam_pattern} shows the comparison results with $N=8,~16,~32,~64$, and from this figure we can find that the beam gains are significant along the desired user directions, and the beam pattern designed are close to the ideal beam pattern along the user directions, which demonstrates that the solution of the beamforming sub-problem is reasonable.

In addition to the beam pattern comparison, we also compare the user beam gains with varying number of antennas in Fig. \ref{fig:beam_gain_N}, where the parameter settings are the same as those in Fig. \ref{fig:beam_pattern}. From Fig. \ref{fig:beam_gain_N}, we can observe that the designed gain of User 2 is equal to the ideal and there is a small gap between the designed user gain and the ideal beam gain for User 1 (as well as the sum beam gain). This is because the designed beam pattern has side lobes which reduces the gains along the User 1 directions. In comparison, an ideal beam pattern does not have side lobes. Fortunately, the gap increases slowly as $N$ increases when $N\leq 40$, and almost does not increase when $N>40$, which shows that the proposed beamforming method behaves robust against the number of antennas.

Fig. \ref{fig:relative_error_avr} shows the average relative gain errors of User 1, User 2 and the sum gain versus the ideal/desired beam gains. The parameter settings are $|\lambda_1|=0.9$, $|\lambda_2|=0.4$. The AoAs of Users $\Omega_1$ and $\Omega_2$ randomly range in $[-1,1]$ with uniform distribution, and there is a constraint $2/N<|\Omega_1-\Omega_2|<(2-2/N)$ because the width of beam gains we designed is $2/N$ in general. The desired beam gains are $c_1^\star=2N/3$ and $c_2^\star=(N-c_1^\star/|\lambda_1|^2)|\lambda_2|^2$, where $N$ is the number of antennas at the BS. Each point in Fig. \ref{fig:relative_error_avr} is the average performance based on $10^3$ beamforming realizations. We find that the relative gain error of User 2 is near zero, which shows that the beamforming setting is almost ideal for User 2. The relative gain error of User 1 is roughly around 0.1, and the relative error of sum gain is no more than 0.1,and they increase slowly as $N$ increases when $N \leq 56$, and almost does not increase when $N>40$. This result not only demonstrates again that the proposed beamforming method behaves robust against the number of antennas, but also shows the rational of Lemma 1, i.e., the sum beam gain can be roughly seen as a constant versus $N$.

 \begin{figure}[t]
\begin{center}
  \includegraphics[width=\figwidth cm]{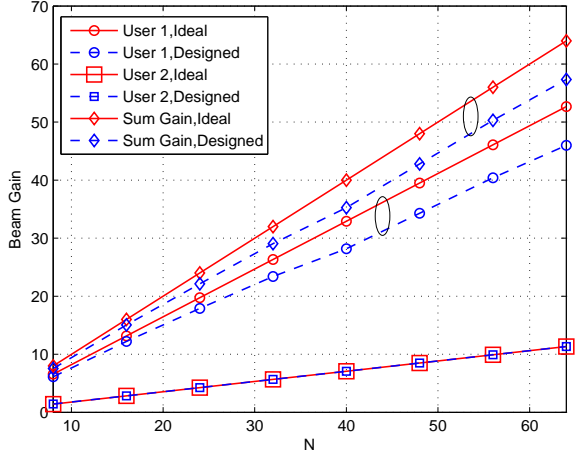}
  \caption{Comparison of user beam gains between the ideal beam gain and the designed beam gain, where the sum gain refers to the summation of the beam gains of User 1 and User 2.}
  \label{fig:beam_gain_N}
\end{center}
\end{figure}

 \begin{figure}[t]
\begin{center}
  \includegraphics[width=\figwidth cm]{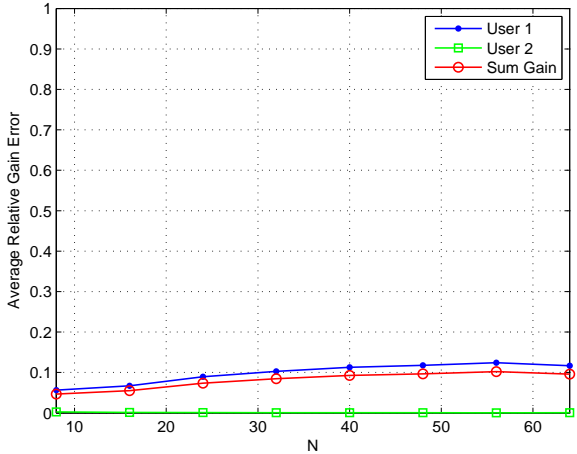}
  \caption{Average relative gain errors versus the ideal beam gains of User 1, User 2 and the sum gain.}
  \label{fig:relative_error_avr}
\end{center}
\end{figure}

The above evaluations show that the solution of the beamforming sub-problem is reasonably close to the ideal one. Next, we evaluate the overall performance. Fig. \ref{fig:achievable_r} shows the comparison between the performance bound and the designed achievable rates with varying rate constraint. The performance bound refers to the achievable rate obtained by solving only the power control and beam gain allocation sub-problem, i.e., with parameters $\{c_{1}^{\star},c_{2}^{\star},p_{1}^{\star},p_{2}^{\star}\}$, where the beamforming is assumed ideal. The designed performance refers to the achievable rate obtained by solving both the power control and beam gain allocation and beamforming sub-problems, i.e., \eqref{optimal_power} and solution of Problem \eqref{eq_problem_beamforming5}. Relevant parameter settings are $\sigma^2=1$ mW, $P=100$ mW, $N=32$, $|\lambda_1|=0.9$, $|\lambda_1|=0.2$, $\Omega_1=-0.7$, $\Omega_2=0.5$. From Fig. \ref{fig:achievable_r} we can find that the designed achievable rates are close to the ideal achievable rates for both User 1 and User 2, as well as the sum rate, which demonstrates that the proposed solution to the original problem is rational and effective, i.e., it can achieve near-optimal performance. On the other hand, we can find that most beam gain is allocated to User 1, which has the better channel condition, so as to optimize the sum rate. Only necessary beam gain is allocated to User 2 to satisfy the rate constraint. That is why User 2 always achieves an achievable rate equal to the rate constraint.

 \begin{figure}[t]
\begin{center}
  \includegraphics[width=\figwidth cm]{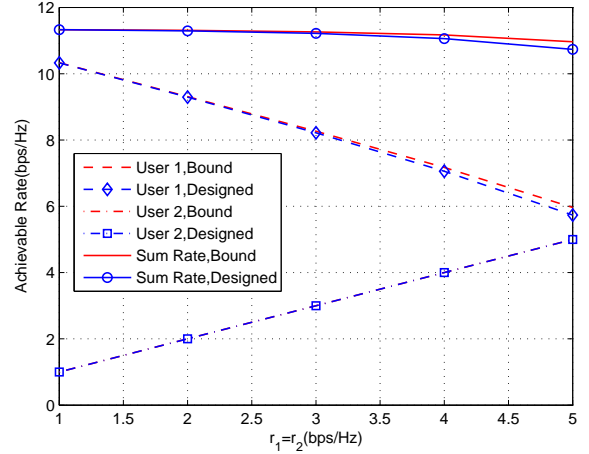}
  \caption{Comparison between the performance bound and the designed achievable rates with varying rate constraint.}
  \label{fig:achievable_r}
\end{center}
\end{figure}

Fig. \ref{fig:achievable_P} shows the comparison between the performance bound and the designed achievable rates with varying maximal power to noise ratio. Relevant parameter settings are $N=32$, $|\lambda_1|=0.9$, $|\lambda_2|=0.2$, $\Omega_1=-0.7$, $\Omega_2=0.5$, $r_1=r_2=3$ bps/Hz. From this figure we can observe the similar results as those from Fig. \ref{fig:achievable_r}, i.e., the designed achievable rates are close to the ideal achievable rates for both User 1 and User 2, as well as the sum rate, and most beam gain is allocated to User 1 to optimize the sum rate, while only necessary beam gain is allocated to User 2 to satisfy the rate constraint.

\begin{figure}[t]
\begin{center}
  \includegraphics[width=\figwidth cm]{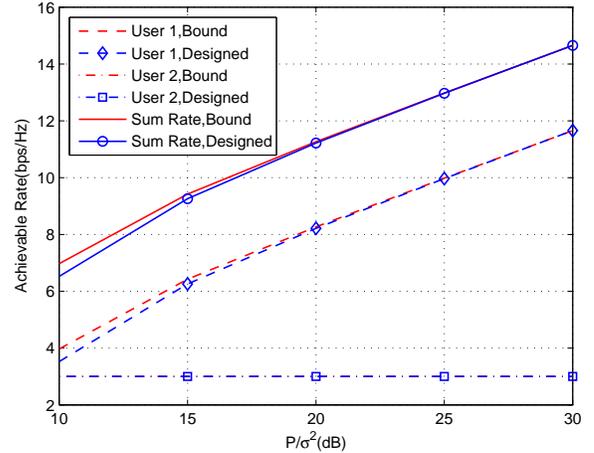}
  \caption{Comparison between the performance bound and the designed achievable rates with varying maximal power to noise ratio.}
  \label{fig:achievable_P}
\end{center}
\end{figure}

Figs. \ref{fig:comparison_r} and \ref{fig:comparison_P} show the comparison results of sum rate between theoretical mmWave NOMA, practical mmWave NOMA and OMA with varying rate constraint and varying maximal power to noise ratio, respectively, where $N=32$ and $L_1=L_2=L=4$. User 1 has a better channel condition than User 2, i.e., the average power ratio of them is $(1/0.3)^2$. For Fig. \ref{fig:comparison_r}, $\frac{P}{\sigma^2}=25$ dB, while for Fig. \ref{fig:comparison_P} $r_1=r_2=2$ bps/Hz. Both LOS and NLOS channel models are considered. For LOS channel, the first path is the LOS path, which has a constant power, i.e., $|\lambda_1|=1$ (0 dB), while the coefficients of the other 3 NLOS paths, i.e., $\{\lambda_i\}_{i=2,3,4}$, obey the complex Gaussian distribution with zero mean, and each of them has an average power of -15 dB. For the NLOS channel, the 4 paths are all NLOS paths with zero-mean complex Gaussian distributed coefficients, and each of them has an average power of $1/\sqrt{L}$. Each point in Figs. \ref{fig:comparison_r} and \ref{fig:comparison_P} is the average performance based on $10^3$ channel realizations. With each channel realization, the optimal parameters are obtained by the proposed solution, and the theoretical/practical performances are obtained by computing the sum rates with the effective/original channel. The performance of OMA is obtained based on the assumption that the beams gains of User 1 and User 2 are equal, i.e., $N/2$, and the instantaneous signal power for each user is $2P$. From these two figures we can observe that the theoretical performance is very close to the practical performance, which demonstrates the rational of the proposed method. Moreover, the performance of mmWave NOMA is significantly better than that of OMA under both the LOS and NLOS channels.
\begin{figure}[t]
\begin{center}
  \includegraphics[width=\figwidth cm]{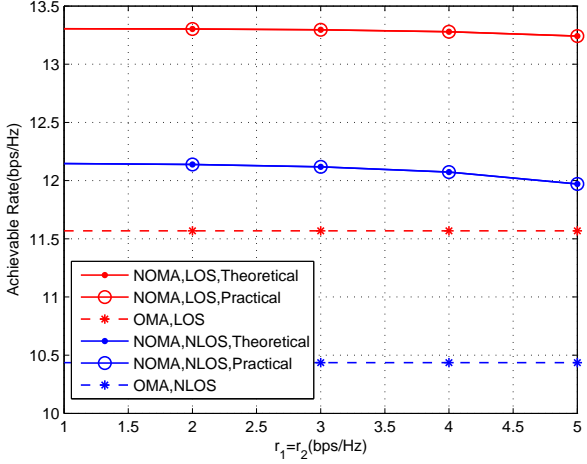}
  \caption{Comparison of sum rate between theoretical mmWave NOMA, practical mmWave NOMA and OMA with varying rate constraint.}
  \label{fig:comparison_r}
\end{center}

\end{figure}
\begin{figure}[t]
\begin{center}
  \includegraphics[width=\figwidth cm]{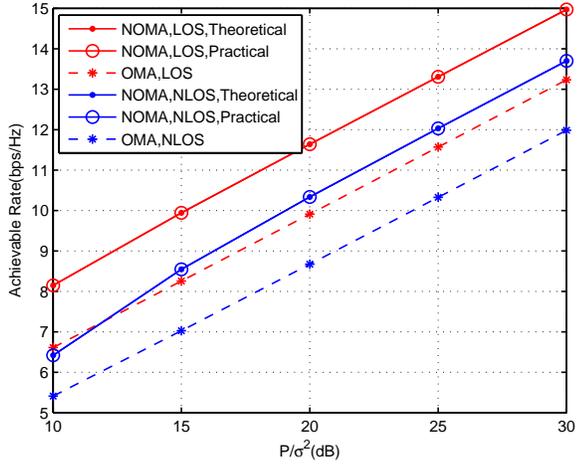}
  \caption{Comparison of sum rate between theoretical mmWave NOMA, practical mmWave NOMA and OMA with varying maximal power to noise
ratio.}
  \label{fig:comparison_P}
\end{center}
\end{figure}

\section{Conclusion}
In this paper we have investigated the problem of how to maximize the sum rate of a 2-user uplink mmWave-NOMA system, where we need to find the beamforming vector to steer to the two users simultaneously subject to an analog beamforming structure at the BS side, and meanwhile select appropriate power for them at the users side. We have proposed a suboptimal solution to this problem, i.e., to decompose the original problem into two sub-problems: one is a power control and beam gain allocation problem, and the other is a beamforming problem under the CM constraint. The original problem can then be solved by solving the two sub-problems. A general system with more users is also studied in this paper, and the basic idea of decomposing the original problem is still workable. Extensive performance evaluations verify the rational of the proposed solution, and demonstrates that the proposed solution can achieve close-to-bound performance, which is distinctively better than OMA.

% or


\appendices
\section{Proof of Theorem1}
Lemma 1 and Lemma 2 are still workable in Case 2. We have $p_{1}=P,~p_{2}=P,~c_{1}=\left|\lambda_{1}\right|^{2}(N-\frac{c_{2}}
{\left|\lambda_{2}\right|^{2}})$. Substituting them into Problem \eqref{eq_problem_sub1_case2}, there is only one independent variable $c_{2}$ now. Hence, we can transform the problem as
\begin{equation}\label{eq_problem2_case2}
\begin{aligned}
\mathop{\mathrm{Maximize}}\limits_{c_{2}}~~~~ &\log_{2}(1+ \frac{(\left|\lambda_{1}\right|^{2}N-(\frac{\left|\lambda_{1}\right|^{2}}{\left|\lambda_{2}\right|^{2}}-1)c_{2})P}{\sigma^{2}})\\
\mathrm{Subject~ to}~~~ &\log_{2}(1+ \frac{\left|\lambda_{1}\right|^{2}(N-\frac{c_{2}}
{\left|\lambda_{2}\right|^{2}})P}{\sigma^{2}}) \geq r_{1}\\
&\log_{2}(1+ \frac{c_{2}P}{\left|\lambda_{1}\right|^{2}(N-\frac{c_{2}}
{\left|\lambda_{2}\right|^{2}})P+\sigma^{2}}) \geq r_{2} \\
\end{aligned}
\end{equation}

Similar to Case 1, the objective function in \eqref{eq_problem2_case2} is monotonically decreasing for $c_{2}$, so the infimum of $c_{2}$ is optimal. Furthermore, $R_1$ is decreasing for $c_2$ and $R_2$ is increasing for $c_2$. The lowerbound of $c_{2}$ is depended on the second constraint $R_2 \geq r_2$ of Problem \eqref{eq_problem2_case2}.
\begin{equation}
\begin{aligned}
\log_{2}(1+ \frac{c_{2}P}{\left|\lambda_{1}\right|^{2}(N-\frac{c_{2}}
{\left|\lambda_{2}\right|^{2}})P+\sigma^{2}}) \geq r_{2} \\
\Leftrightarrow c_{2} \geq \frac{(\left|\lambda_{1}\right|^{2}NP+\sigma^{2})(2^{r_{2}}-1)}{(1+\frac{\left|\lambda_{1}\right|^{2}}{\left|\lambda_{2}\right|^{2}}(2^{r_{2}}-1))P}
\end{aligned}
\end{equation}

The lowerbound of $c_{2}$ can be obtained when $R_2=r_2$ in both cases. Denote them as $c_2^{(1)}$ and $c_2^{(2)}$ respectively, and we have
\begin{equation}
\begin{aligned}
&\log_{2}(1+ \frac{c_{2}^{(1)}P}{\sigma^{2}}) = r_{2}\\
&= \log_{2}(1+ \frac{c_{2}^{(2)}P}{\left|\lambda_{1}\right|^{2}(N-\frac{c_{2}^{(2)}}
{\left|\lambda_{2}\right|^{2}})P+\sigma^{2}})\\
&\leq \log_{2}(1+ \frac{c_{2}^{(2)}P}{\sigma^{2}})\\
&\Leftrightarrow c_{2}^{(1)} \leq c_{2}^{(2)}
\end{aligned}
\end{equation}
As we have mentioned before, the objective function in Problem \eqref{eq_problem2_case1} and \eqref{eq_problem2_case2} is identical, which is monotonically decreasing for the variable $c_2$. Then we can conclude that the optimal solution of Case 1 is better than Case 2, because $c_{2}^{(1)} \leq c_{2}^{(2)}$.

\section{Proof of Theorem2}

 \begin{figure*}[t]
\begin{center}
  \includegraphics[width=17 cm]{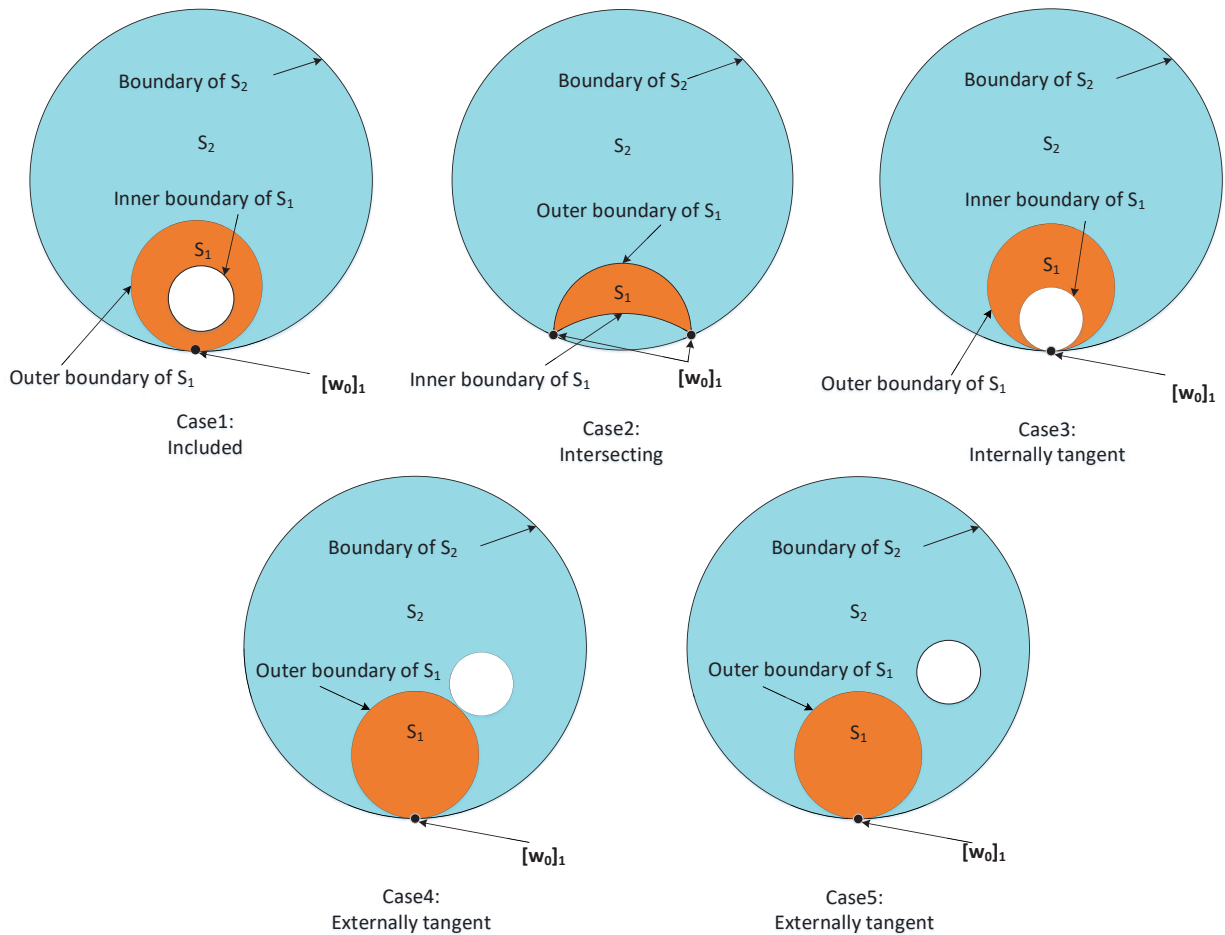}
  \caption{Illustration of the relative position relation between $S_{1}$ and $S_{2}$. On one hand, $S_{1} \subseteq S_{2}$; On the other hand, $[{\mathbf{w}}_{0}]_1$ is the intersection between $S_1$ and the boundary of $S_2$. Thus, no matter what the shape of $S_1$ is, $[{\mathbf{w}}_{0}]_1$ must be located in the outer boundary of $S_1$. }
  \label{fig:case123}
\end{center}
\end{figure*}

Let $\mathbf{w}_{0}$ represent the optimal solution of Problem \eqref{eq_problem_beamforming3}, and
\begin{equation}
\left\{\begin{aligned}
&{\mathbf{a}}_{1}^{\rm{H}}{\mathbf{w}}_{0}=d_{1}e^{\theta_{1}j}\\
&{\mathbf{a}}_{2}^{\rm{H}}{\mathbf{w}}_{0}=d_{2}e^{\theta_{2}j}
\end{aligned}\right.
\end{equation}
where $d_{i}$ and $\theta_{i}$ denote the modulus and phase of ${\mathbf{a}}_{i}^{\rm{H}}{\mathbf{w}}_{0}$, respectively. We will show $|[{\mathbf{w}}_{0}]_i|=\frac{1}{\sqrt{N}}$ for $i=1,2,...,N$. For this sake, we will only prove that $|[{\mathbf{w}}_{0}]_1|=\frac{1}{\sqrt{N}}$ in detail, while $|[{\mathbf{w}}_{0}]_i|=\frac{1}{\sqrt{N}}$ for $i=2,3,...,N$ can be proven similarly. As the modulus of $[{\mathbf{a}}_{i}]_1 (i=1,2)$ is 1, we have $|[{\mathbf{a}}_{i}]_1[{\mathbf{w}}_{0}]_1|=|[{\mathbf{w}}_{0}]_1|\triangleq l$.

Denote
\begin{equation}
\left\{\begin{aligned}
&[{\mathbf{a}}_{1}^{\rm{H}}]_{1}[{\mathbf{w}}_{0}]_{1}=le^{\mu_{1}j}\\
&[{\mathbf{a}}_{2}^{\rm{H}}]_{1}[{\mathbf{w}}_{0}]_{1}=le^{\mu_{2}j}
\end{aligned}\right.
\end{equation}
and
\begin{equation}
\left\{\begin{aligned}
&\sum\limits_{k=2}^{N}[{\mathbf{a}}_{1}^{\rm{H}}]_{k}[{\mathbf{w}}_{0}]_{k}=b_{1}e^{\nu_{1}j}\\
&\sum\limits_{k=2}^{N}[{\mathbf{a}}_{2}^{\rm{H}}]_{k}[{\mathbf{w}}_{0}]_{k}=b_{2}e^{\nu_{2}j}
\end{aligned}\right.
\end{equation}

Obviously, $le^{\mu_{i}j}+b_{i}e^{\nu_{i}j}=d_{i}e^{\theta_{i}j}$. Note that the phase difference between $[{\mathbf{a}}_{1}^{\rm{H}}]_{1}[{\mathbf{w}}_{0}]_{1}$ and $[{\mathbf{a}}_{2}^{\rm{H}}]_{1}[{\mathbf{w}}_{0}]_{1}$, i.e., $({\mu}_{2}-{\mu}_{1})$, does not dependent on $[{\mathbf{w}}_{0}]_{1}$. Next, we will show that the optimal $[{\mathbf{w}_{0}}]_1$ must be on the constraint boundary $|[{\mathbf{w}}_{0}]_1|=\frac{1}{\sqrt{N}}$.

For the constraints in Problem \eqref{eq_problem_beamforming3}. For fixed $[{\mathbf{w}}_{0}]_{k}~(k=2,3,\cdots,N)$, the constraints for $[{\mathbf{w}}_{0}]_{1}$ are
\begin{equation}\label{crescent}
\left\{\begin{aligned}
&|{\mathbf{a}}_{2}^{\rm{H}}{\mathbf{w}}_{0}|=|l e^{\mu_{2}j}+b_{2}e^{\nu_{2}j}| \geq g\\
&|[{\mathbf{w}}_{0}]_{1}|=l \leq \frac{1}{\sqrt{N}}
\end{aligned}\right.
\end{equation}

Consider the above variables in the polar coordinate system, where the constraints \eqref{crescent} denote a feasible region in the 2-dimensional plane. $|l e^{\mu_{2}j}+b_{2}e^{\nu_{2}j}| \geq g$ is the outside part of a circle and $l \leq \frac{1}{\sqrt{N}}$ is the inside part of a circle. Hence, the feasible region of \eqref{crescent}, denoted by $S_{1}$, is a closed set with two boundaries\footnote{$S_1$ is not empty because there is at least one point, $[{\mathbf{w}}_{0}]_{1}$.}. One is the equation $|l e^{\mu_{2}j}+b_{2}e^{\nu_{2}j}| = g$. We define it as the inner boundary of $S_{1}$. The other is the equation $l = \frac{1}{\sqrt{N}}$. We define it as the outer boundary of $S_{1}$. The shape of $S_{1}$ depends on the relative position relation between the two circles, i.e., included, intersecting, internally tangent, externally tangent and separate, which are shown in Fig. \ref{fig:case123}, where $S_2$ is defined below.

It is assumed that the objective function of Problem \eqref{eq_problem_beamforming3} is maximum at the point $[{\mathbf{w}_{0}}]_1$, which is described by
\begin{equation}\label{disc}
\begin{aligned}
&|{\mathbf{a}}_{1}^{\rm{H}}{\mathbf{w}}_{0}|=|le^{\mu_{1}j}+b_{1}e^{\nu_{1}j}| = d_{1}\\
\Leftrightarrow &|le^{\mu_{1}j+({\mu}_{2}-{\mu}_{1})j}+b_{1}e^{\nu_{1}j+({\mu}_{2}-{\mu}_{1})j}| = d_{1}\\
\Leftrightarrow &|le^{\mu_{2}j}+b_{1}e^{\nu_{1}j+({\mu}_{2}-{\mu}_{1})j}| = d_{1}
\end{aligned}
\end{equation}
where $b_{1}$, $\nu_{1}$ and $({\mu}_{2}-{\mu}_{1})$ are constant. In other words, $d_1$ is the maximum distance from the point $-b_{1}e^{\nu_{1}j+({\mu}_{2}-{\mu}_{1})j}$ to the region $S_{1}$. If we draw a circle centered at the point $-b_{1}e^{\nu_{1}j+({\mu}_{2}-{\mu}_{1})j}$ with the radius of $d_1$, then $S_1$ is certainly located inside of this circle. Otherwise the the point outside of this circle is optimal, which is contradictory to the assumption. The inside part of this circle is described by $|le^{\mu_{2}j}+b_{1}e^{\nu_{1}j+({\mu}_{2}-{\mu}_{1})j}| \leq d_{1}$ and denoted by $S_{2}$ (see also Fig. \ref{fig:case123}). In particular, we define the equation $|le^{\mu_{2}j}+b_{1}e^{\nu_{1}j+({\mu}_{2}-{\mu}_{1})j}| = d_{1}$ as the boundary of $S_{2}$. Then we have $S_{1} \subseteq S_{2}$. It can be seen that ${\mathbf{w}}_{0}$ is located in the outer boundary of $S_{1}$ in Fig. \ref{fig:case123}, no matter what the shape of $S_{1}$ is. Thus, $|[{\mathbf{w}}_{0}]_1|=\frac{1}{\sqrt{N}}$.

\section{Proof of Theorem3}
It is obvious that $p_{i}^{\star}=P~(i=1,2)$ satisfies the power constraint for User $i$. And the CM constraint for the beamforming vector $\mathbf{w}^\circ$ is also considered in Problem \eqref{eq_problem_beamforming1}. Thus we just need to verify that
\begin{equation}
\left\{
\begin{aligned}
&R_{1}^{\star} \geq r_{1}\\
&R_{2}^{\star} \geq r_{2} \\
\end{aligned}
\right.
\end{equation}
where $R_{i}^{\star}~(i=1,2)$ is the achievable rate of User $i$ under proposed solution $\{p_{1}^{\star},p_{2}^{\star},\mathbf{w}^\circ\}$.

On one hand, we have
\begin{equation}
\begin{aligned}
R_{2}^{\star}&=\log_{2}(1+ \frac{\left |\mathbf{h}_{2}^{\rm{H}}\mathbf{w}^\circ \right |^{2}p_{2}^{\star}}{\sigma^{2}}) \\
&\geq \log_{2}(1+ \frac{c_{2}^{\star}P}{\sigma^{2}})\\
&=r_{2}
\end{aligned}
\end{equation}

One the other hand, in Problem \eqref{eq_problem2_case1}, the optimal solution is located in the boundary of $R_{2}=r_{2}$, which means only necessary beam gain is allocated to User 2 to satisfy the minimum rate constraint and the rest of beam gain is all allocated to User 1. Similar in \eqref{eq_problem_beamforming1}, we try to maximize the beam gain of User 1 while the beam gain of User 2 just insures the minimum gain to satisfy the rate constraint. Thus the combination of Problem \eqref{eq_problem2_case1} and \eqref{eq_problem_beamforming1} is equivalent to
\begin{equation}\label{eq_problem_com}
\begin{aligned}
\mathop{\mathrm{Maximize}}\limits_{p_1,p_2,{\mathbf{w}}}~~~~ &R_{1}\\
\mathrm{Subject~ to}~~~ &R_{2} \geq r_{2}\\
&0 \leq p_{1},p_{2} \leq P \\
&|[{\mathbf{w}}]_k|=\frac{1}{\sqrt{N}},~k=1,2,...,N\\
\end{aligned}
\end{equation}

Assume that $R_{1}^{\star} < r_{1}$, which means that under the constraints of Problem \eqref{eq_problem_com}, the maximum value of $R_{1}$ is smaller than $r_{1}$. In other words, the constraint $R_{1} \geq r_{1}$ in Problem \eqref{eq_problem} cannot be feasible. The feasible region of Problem \eqref{eq_problem} is empty. However, we have assumed that the feasible region of Problem \eqref{eq_problem} is not empty in Theorem 3, which is contradictory. Thus there must be $R_{1}^{\star} \geq r_{1}$.

%\bibliographystyle{IEEEtran} % use IEEEtran.bst style
%\bibliography{IEEEabrv,Xiao60GHz,Xiao5GnNOMA}


% that's all folks%
\end{document}